\pdfoutput=1
\documentclass[11pt,a4paper]{article}
\usepackage{amsthm,color,ifpdf,latexsym,graphicx,subcaption,url}

\usepackage[cp1251,utf8]{inputenc}
\usepackage[T1,T2A]{fontenc}
\usepackage[russian,english]{babel}

\usepackage[margin=1in]{geometry}
\usepackage[labelfont=bf]{caption}
\usepackage{amsfonts}
\title{Polyhedral Characterization of Reversible Hinged Dissections}
\author{%
  Jin Akiyama%
    \thanks{Tokyo University of Science}
\and
  Erik D. Demaine%
    \thanks{CSAIL, Massachusetts Institute of Technology}
\and
  Stefan Langerman%
    \thanks{Directeur de Recherches du F.R.S-FNRS, Universit\'e Libre de Bruxelles}
}
\date{}

\newif\ifabstract
\abstracttrue
\abstractfalse
\newif\iffull
\ifabstract \fullfalse \else \fulltrue \fi

\usepackage{hyperref}
\hypersetup{breaklinks,bookmarks,bookmarksnumbered,bookmarksopen,bookmarksopenlevel=2}
{\makeatletter \hypersetup{pdftitle={\@title}}}

{\makeatletter
 \gdef\xxxmark{%
   \expandafter\ifx\csname @mpargs\endcsname\relax 
     \expandafter\ifx\csname @captype\endcsname\relax 
       \marginpar{xxx}
     \else
       xxx 
     \fi
   \else
     xxx 
   \fi}
 \gdef\xxx{\@ifnextchar[\xxx@lab\xxx@nolab}
 \long\gdef\xxx@lab[#1]#2{\textbf{[\xxxmark #2 ---{\sc #1}]}}
 \long\gdef\xxx@nolab#1{\textbf{[\xxxmark #1]}}
}

{\makeatletter \gdef\fps@figure{!htbp}}

\let\realbfseries=\bfseries
\def\bfseries{\realbfseries\boldmath}

\newtheorem{theorem}{Theorem}[section]
\newtheorem{lemma}[theorem]{Lemma}

\let\epsilon=\varepsilon

\begin{document}
\maketitle

\begin{abstract}
  We prove that two polygons $A$ and $B$ have a reversible hinged dissection
  (a chain hinged dissection that reverses inside and outside boundaries
  when folding between $A$ and~$B$) if and only if $A$ and $B$ are two
  noncrossing nets of a common polyhedron.  Furthermore,
  \emph{monotone} reversible hinged dissections (where all hinges rotate in the same
  direction when changing from $A$ to~$B$) correspond exactly to
  noncrossing nets of a common convex polyhedron.  By envelope/parcel magic,
  it becomes easy to design many
  hinged dissections.
\end{abstract}

\section{Introduction}

Given two polygons $A$ and $B$ of equal area, a \emph{dissection}
\cite{Frederickson-1997}
is a decomposition of $A$ into pieces that can be re-assembled (by
translation and rotation) to form $B$. In a (chain) \emph{hinged} dissection
\cite{Frederickson-2002},
the pieces are hinged together at their corners to form a chain,
which can fold into both $A$ and~$B$,
while maintaining connectivity between pieces at the hinge points.
Figure~\ref{dudeney} shows a famous example by Dudeney \cite{Dudeney-1902-hinged}.
Many known hinged dissections (including Figure~\ref{dudeney}) are
\emph{reversible} (originally called \emph{Dudeney dissection}
\cite{Akiyama-Nakamura-1998}), meaning that the
outside boundary of $A$ goes inside of $B$ after the reconfiguration,
while the portion of the boundaries of the dissection inside of $A$
become the exterior boundary of $B$.
In particular, the hinges must all be on the boundary of both $A$ and~$B$,
in the opposite counterclockwise order.
We thus view reversible hinged dissections as a \emph{cyclic} chain of pieces
and hinges, because the choice of the hinge to cut to perform the
reconfiguration is irrelevant: the two endpoints of the chain meet
in both configurations.
Other papers \cite{Akiyama-Seong-2013,Akiyama-Langerman-Matsunaga-2015}
call the pair $A,B$ of polygons (instead of the hinged dissection)
\emph{reversible}.

\begin{figure}
  \centering
  $$
  \vcenter{\hbox{\includegraphics[width=0.33\linewidth]{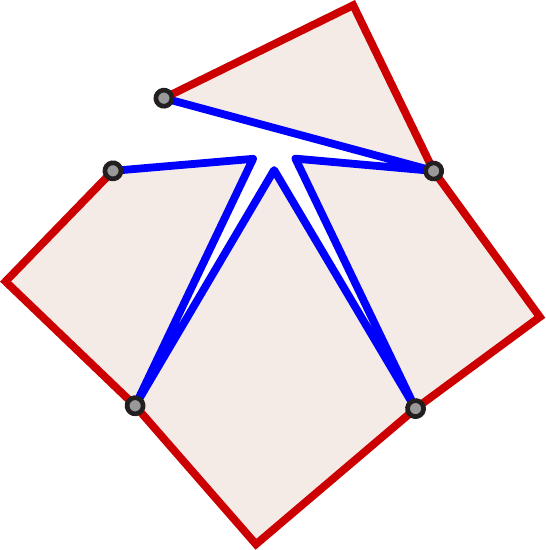}}}
  \quad\leftrightarrow\quad
  \vcenter{\hbox{\includegraphics[width=0.33\linewidth]{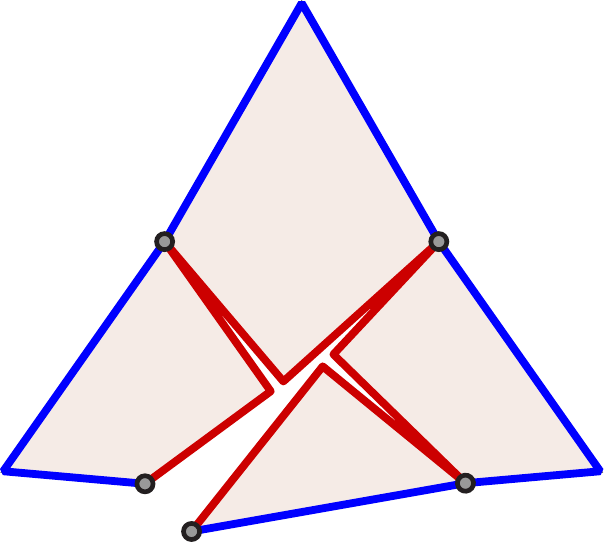}}}
  $$
  \caption{Dudeney's reversible hinged dissection of a square to an equilateral triangle \cite{Dudeney-1902-hinged}.}
  \label{dudeney}
\end{figure}

Without the reversibility restriction,
Abbott et al.~\cite{Abbott-Abel-Charlton-Demaine-Demaine-Kominers-2012}
showed that any two polygons of same area have a hinged dissection.
Properties of reversible pairs of polygons were studied by Akiyama et
al.~\cite{Akiyama-Nakamura-1998,Akiyama-Seong-2013}.
A recent paper \cite{Akiyama-Langerman-Matsunaga-2015} described the
\emph{parcel magic} method to generate reversible hinged dissections.
This method works by cutting open, unfolding, and flattening a polyhedron in
two different ways such that the cut trees of the two unfoldings do not
intersect. The special case when the polyhedron is a dihedron (flat doubly
covered polygon) is called \emph{envelope magic}.
The purpose of this paper is to formalize this method and show that
this characterization is in some sense complete, that is, that every
reversible hinged dissection can be constructed this way.

More precisely, we show the following three results:

\begin{enumerate}
\item Two polygons $A,B$ have a reversible hinged dissection
  if and only if
  $A$ and $B$ are two noncrossing nets of a common polyhedron
  (Theorem~\ref{thm:nonmonotone}).
\item Two polygons $A,B$ have a \emph{monotone} reversible hinged dissection
  (where all the turn angles of all hinges increase from $A$ to $B$)
  if and only if
  $A$ and $B$ are two noncrossing nets of a common \emph{convex} polyhedron
  (Theorem~\ref{thm:monotone}).
\item Two polygons $A,B$ have a \emph{nondegenerate} reversible hinged
  dissection (where each hinge touches just its two incident pieces),
  if and only if
  $A$ and $B$ are two noncrossing nets of a common \emph{convex} polyhedron
  that have only one cut incident to each polyhedron vertex
  (Theorem~\ref{thm:nondegenerate}).
\end{enumerate}

\section{Noncrossing Nets}
The heart of our results is a lemma about circumnavigating a polyhedron
between two noncrossing cut trees of unfoldings.

First we need some terminology.
In this paper, a \emph{polyhedron} is always homeomorphic to a sphere.
An \emph{unfolding} of a polyhedron $P$ cuts the surface of $P$ using a
\emph{cut tree} $T$,%
\footnote{For simplicity, we assume that the edges
  of $T$ are drawn using segments along the surface of $P$, and that
  vertices of degree 2 can be used in $T$ to draw any polygonal path.}
spanning all vertices of $P$, such that the cut surface $P \setminus T$
can be unfolded into the plane without overlap by opening all dihedral
angles between the (possibly cut) faces.
The planar polygon that results from this
unfolding is called a \emph{net} of~$P$.
Two trees $T_1$ and $T_2$ drawn on a surface
are \emph{noncrossing} if pairs of edges of $T_1$ and $T_2$
intersect only at common endpoints and, for any vertex $v$ of both $T_1$ and
$T_2$, the edges of~$T_1$ (respectively,~$T_2$) incident to $v$ are contiguous
in clockwise order around~$v$. Two nets of a common polyhedron
are noncrossing if their cut trees are noncrossing.

\begin{lemma} \label{cycle}
  Let $T_1, T_2$ be noncrossing trees drawn on a polyhedron $P$,
  each of which spans all vertices of~$P$.
  Then there is a cycle $C$ passing through all vertices of $P$
  such that $C$ separates the edges of $T_1$ from edges of $T_2$, i.e.,
  the (closed) interior (yellow region, see Figure~\ref{example}) of $C$ includes all edges of
  $T_1$ and the (closed) exterior of $C$ includes all edges of~$T_2$.
  Furthermore, all such cycles visit the vertices of $P$ in the same order.
\end{lemma}

\begin{figure}
  \centering
  \includegraphics[width=0.4\linewidth]{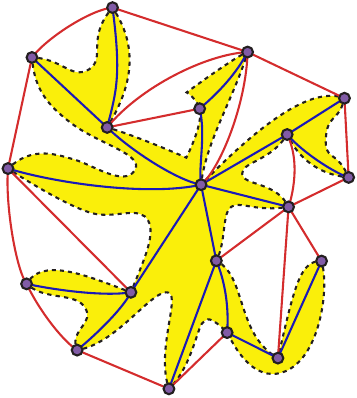}
  \caption{Example of Lemma~\ref{cycle}.
    The edges of $T_1, T_2$ are colored blue, red, respectively.}
  \label{example}
\end{figure}

\iffull
\begin{proof}
Refer to Figure~\ref{steps}.
Let $G$ be the union of $T_1$ and $T_2$.
Because $T_1$ and $T_2$ are noncrossing, $G$ is a planar graph.
Let $\alpha>0$ be a third of the smallest angle between any two incident edges
in~$G$, or $90^\circ$, whichever is smaller. 
Let $\varepsilon>0$ be a third of the smallest distance\footnote{We
  here refer to the geodesic distance on the surface of $P$.}
 between
any edge of $G$ and a vertex not incident to that edge.
View each edge of $G$ as the union of two directed half-edges.
For every half-edge $u,v$ in $G$, its \emph{sidewalk} is a
polygonal path $u,p_{uv},q_{uv},v$ composed of three segments such that 
\begin{enumerate}
\item the counterclockwise angles $\angle p_{uv},u,v$ and $u,v,\angle q_{uv}$
      are both $\alpha$
      (placing $p_{uv}$ and $q_{uv}$ on the left of the directed line $u,v$);
      and
\item both $p_{uv}$ and $q_{uv}$ are at distance $\varepsilon$
      from the segment $u,v$.
\end{enumerate}

\begin{figure}
  \centering
  \subcaptionbox{\label{tree}Tree $T_1$ (purple) and its local interaction with tree $T_2$ (red)}{\includegraphics[width=0.31\linewidth]{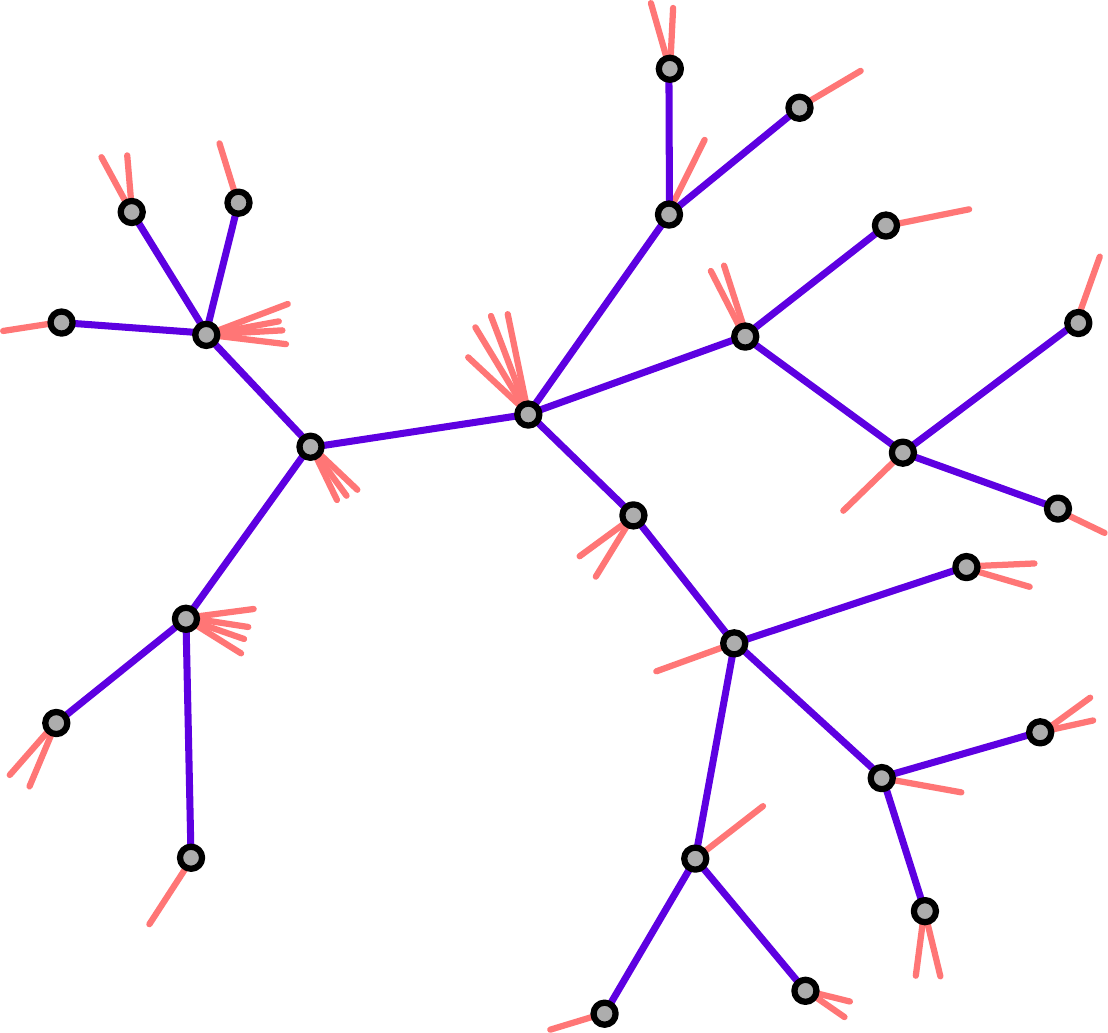}}%
  \hfill
  \subcaptionbox{\label{sidewalks}Sidewalks}{\includegraphics[width=0.31\linewidth]{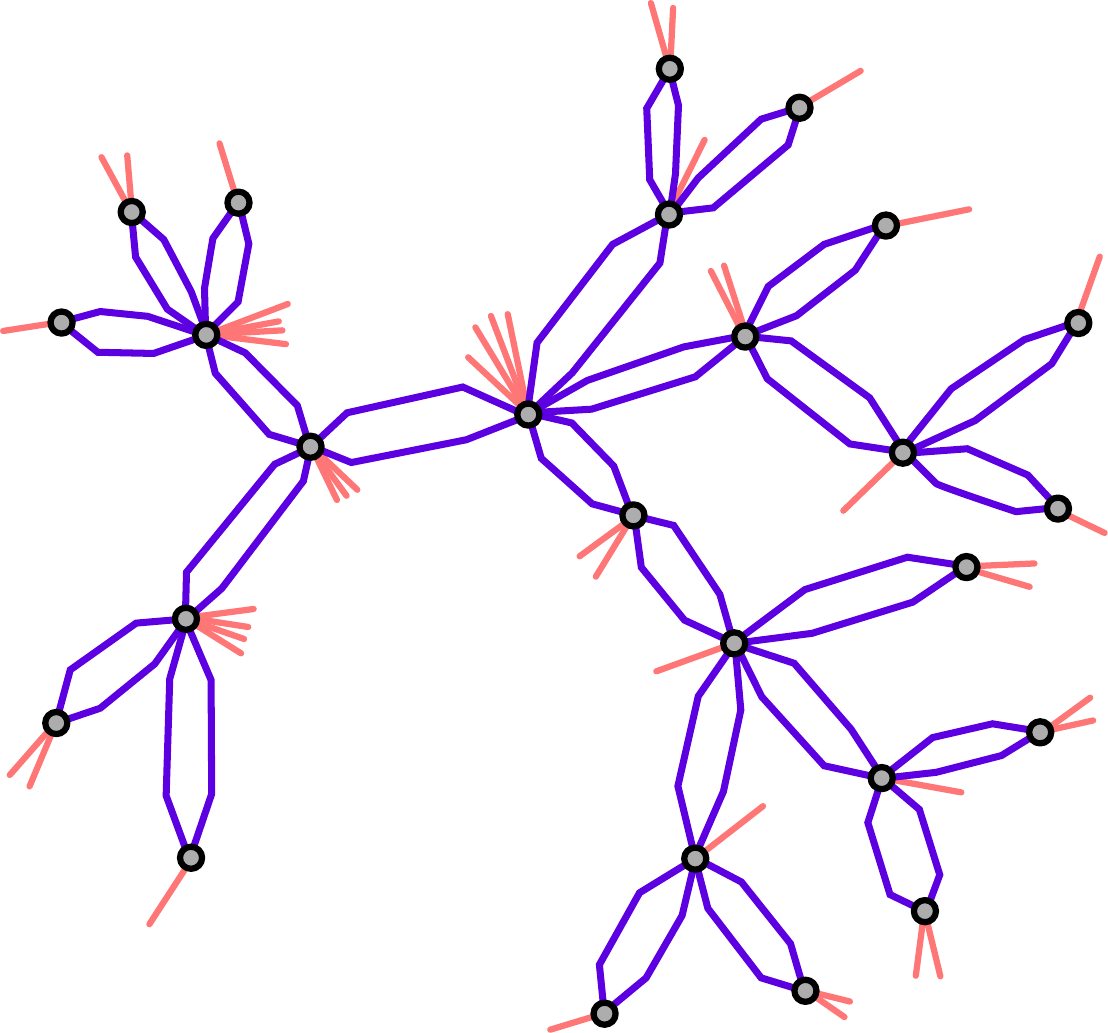}}%
  \hfill
  \subcaptionbox{\label{crosswalks}Crosswalks}{\includegraphics[width=0.31\linewidth]{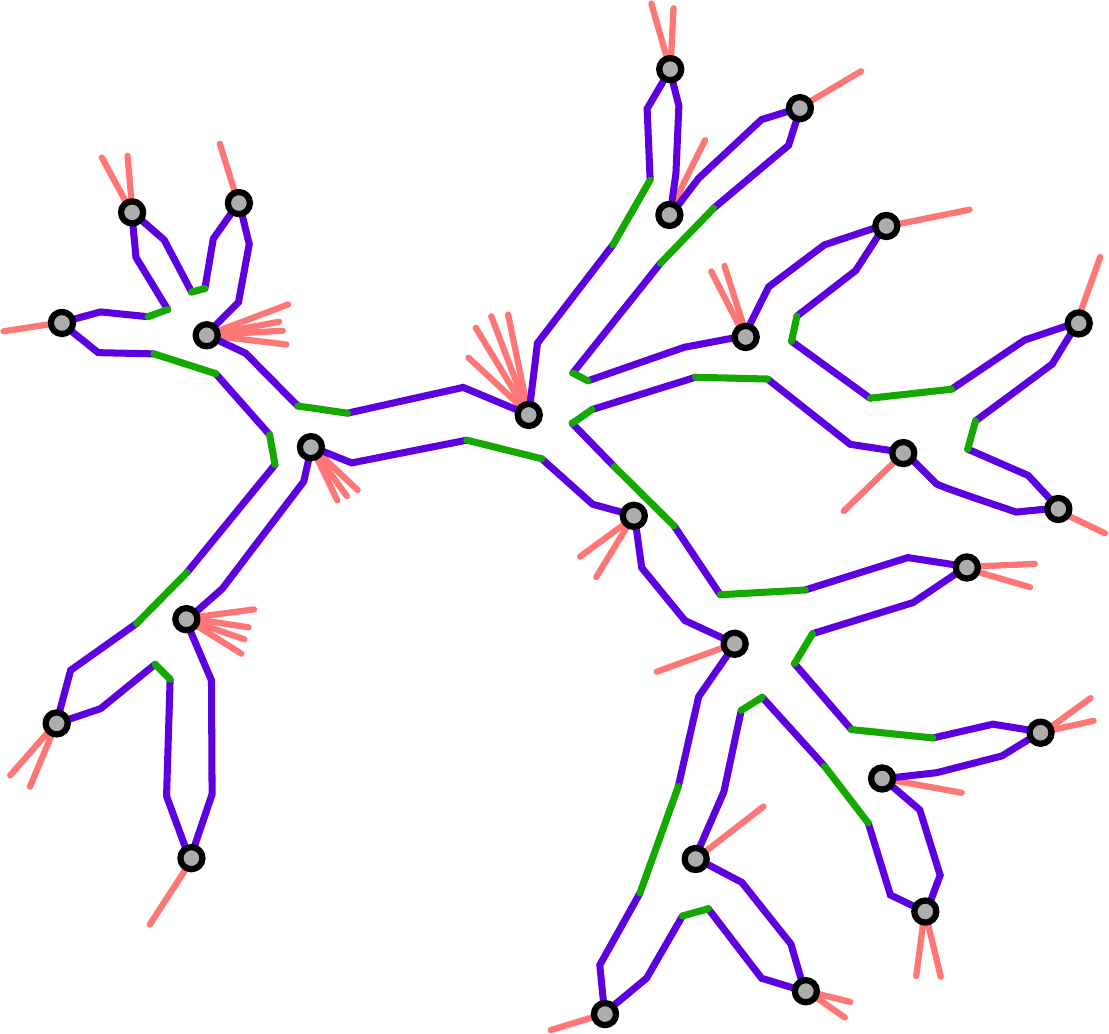}}%
  \caption{Proof of Lemma~\ref{cycle}.}
  \label{steps}
\end{figure}

By construction, $u,v$ is the unique closest edge of $G$ from any point on
its sidewalk. Thus, no two distinct sidewalks intersect and 
sidewalks do not intersect edges of~$G$.

Construct an Euler tour of $T_1$ (Figure~\ref{tree}) that is noncrossing and traverses
clockwise around $T_1$, and replace each step from $u$ to $v$
in the tour by the sidewalk of $u,v$. The concatenation of all
these sidewalks forms a clockwise cycle that visits each vertex $v$
as many times as the degree of~$v$ (Figure~\ref{sidewalks}).
For any two consecutive sidewalks
$u,p_{uv},q_{uv},v,p_{vw},q_{vw},w$ where the wedge $u,v,w$ does not contain
an edge of $T_2$ incident to $v$, shorten the walk by using the
\emph{crosswalk} $q_{uv},p_{vw}$ to obtain 
$\ldots, p_{uv},q_{uv},p_{vw},q_{vw},\ldots$,
thereby avoiding a duplicate visit of~$v$.
Because $T_1$ and $T_2$ are noncrossing, all but one of the visits of
each vertex $v$ will be removed by using crosswalks (Figure~\ref{crosswalks}).

The resulting walk is a simple closed Jordan curve $C$ that visits
each vertex of $G$ exactly once. 
Because $C$ does not intersect $T_1$ and $T_2$,
and locally separates the edges of $T_1$ and $T_2$ at each vertex, and
because $T_1$ and $T_2$ are connected, the curve $C$ separates $P$
into two regions, one containing $T_1$ and the other containing $T_2$.

Finally, we show that the order of vertices of $P$ visited
by any such cycle $C$ is determined by $P$, $T_1$, and~$T_2$.
Consider the planar graph $T_1 \cup T_2$ drawn on~$P$.
We claim that every face of this graph consists of at most one path of
edges from $T_1$ and at most one path of edges from~$T_2$.
Otherwise, we would have at least two components of $T_1$ and at least two
components of $T_2$, neither of which could be connected interior to the face
(because the face is empty), and at most one of which could be connected
exterior to the face (by planarity and the noncrossing property),
contradicting that $T_1$ and $T_2$ are both trees.
Therefore, every face with at least one edge from $T_1$ and at least one edge
from $T_2$ locally forces where $C$ must go, connecting the two
vertices with incident face edges from both $T_1$ and~$T_2$.
Every vertex of $P$ has at least one incident edge from each of the spanning
trees $T_1$ and~$T_2$, so has two incident such faces.
In this way, we obtain the forced vertex ordering of~$C$.
\end{proof}
\fi

\section{Reversible Hinged Dissections}
We can now give our first characterization.

\begin{theorem}\label{thm:nonmonotone}
Two polygons $A$ and $B$ have a reversible hinged dissection if and
only if $A$ and $B$ are two noncrossing nets of a common polyhedron.
\end{theorem}

\begin{proof}
To prove one direction (``only if''), it suffices to glue both sides
of the pieces of the dissection as they are glued in both $A$ and $B$
to obtain a polyhedral metric homeomorphic to a sphere.
A result of Burago and Zalgaller \cite{BZ60,BZ95}\footnote{Burago and
  Zalgaller first proved the result for any orientable surface
  \cite{BZ60} but later noted and fixed a flaw in their construction \cite{BZ95}. At the same time, they generalized their
  result to a stronger statement about possibly non-orientable surfaces.} 
shows that any such
metric corresponds to the surface of some (not necessarily unique,
possibly nonconvex) polyhedron \cite{orourke2010folding}.
In the other direction (``if''), we use Lemma~\ref{cycle} to define the
sequence of hinges. Now the cut tree $T_B$ of net $B$ is completely
contained in the net $A$ and determines the hinged dissection.
\end{proof}

\section{Monotone Reversible Hinged Dissections}

Often times, reversible hinged dissections are also
\emph{monotone} meaning that, if we order the hinges in the dissection
counterclockwise around the boundary of~$A$, then the turn angle at every
hinge decreases (i.e., turns to the right) when transforming from $A$ to~$B$.
(This definition is symmetric in $A$ and $B$ because $B$'s counterclockwise
order of the hinges is the reverse of $A$'s counterclockwise order of the
hinges.)
Recall that we view reversible hinged dissections as a cyclic chain,
so the monotonicity definition also measures the change in angle of the hinge
that is cut open in order to perform the reconfiguration but which recombines
into a vertex in both the $A$ and $B$ configurations.
Figure~\ref{dudeney} is monotone, while Figure~\ref{nonmon} shows a
hinged dissection that is reversible but not monotone.

\begin{figure}
  $$
  \vcenter{\hbox{\includegraphics[width=0.33\linewidth]{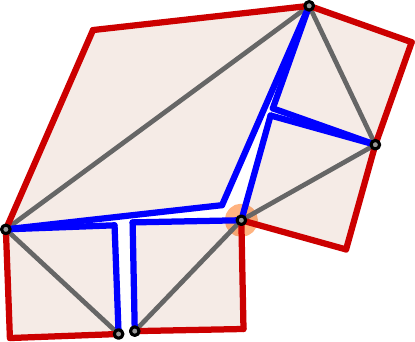}}}
  \quad\leftrightarrow\quad
  \vcenter{\hbox{\includegraphics[width=0.33\linewidth]{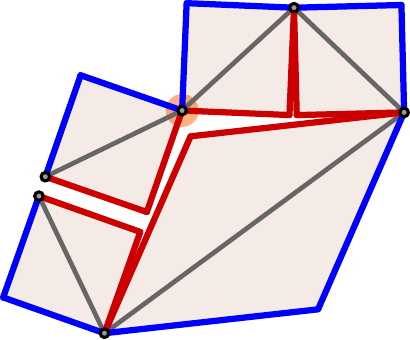}}}
  $$
  \caption{Reversible hinged dissection that is not monotone
           (nor is it nondegenerate), because of the highlighted vertex.}
  \label{nonmon}
\end{figure}

Our second characterization shows that monotonicity in the hinged dissection
is equivalent to convexity of the polyhedron:

\begin{theorem}\label{thm:monotone}
Two polygons $A$ and $B$ have a monotone reversible hinged dissection if and
only if $A$ and $B$ are two noncrossing nets of a common convex polyhedron.
\end{theorem}

\iffull
\begin{proof}
Let $v$ be a hinge of the monotone reversible hinged dissection. Pick
two reference points $v^-$ and $v^+$ in the neighborhood of $v$
and in the pieces before and after hinge $v$, respectively,
in counterclockwise order around the boundary of~$A$.
Let $\alpha_v$ be the counterclockwise angle $\angle v^-,v,v^+$
when the dissection forms polygon $A$, and
let $\alpha'_{v}$ be the same angle when the dissection forms polygon~$B$.
Because the dissection is monotone, $\alpha'_v \geq \alpha_v$ for all $v$.

Just as in Theorem~\ref{thm:nonmonotone}, glue both sides of the
dissection as they are glued in both $A$ and $B$ to obtain a
polyhedral metric homeomorphic to a sphere. Observe now that the total
angle glued at vertex $v$ is exactly 
$\alpha_v+(2\pi-\alpha'_v) \leq 2\pi$.
By Alexandrov's Theorem
\cite[Section~23.3]{Alexandrov-2005,Demaine-O'Rourke-2007},
there exists a unique convex polyhedron (up to rigid transformations)
whose surface has this intrinsic metric.

In the other direction, suppose we have two noncrossing nets of a
convex polyhedron $P$. Use Lemma~\ref{cycle} to find a cycle $C$
separating $T_A$ and $T_B$ on the surface of $P$ and to define the
sequence of hinges, and cut both trees to obtain the dissection.
Pick points $v^-$ and $v^+$ before and after $v$ on $C$ and in the
neighborhood of $v$. Let $\alpha_v$ be the angle $\angle v^-,v,v^+$ in net
$A$ and on the surface of $P$, and $\beta_v$ be the angle
$\angle v^+,v,v^-$ in net $B$ and on the surface of $P$. Because $P$ is
convex, $\alpha_v+\beta_v \leq 2\pi$.
The angle $\alpha'_v = \angle v^-,v,v^+$ when the dissection forms
polygon $B$ is exactly $2\pi-\beta_v \geq \alpha_v$ for every hinge
$v$, and so the dissection is monotone.
\end{proof}
\fi

Two unfoldings of a common convex polyhedron is in some sense the dual notion
of two convex polyhedra with a \emph{common unfolding}, a topic that
has been studied extensively;
see \cite[Section~25.8.3]{Uehara-2016-boxes,Demaine-O'Rourke-2007}.

\section{Nondegenerate Reversible Hinged Dissections}

An interesting special case of a monotone reversible hinged dissection
is when every hinge touches only its two incident pieces in both its
$A$ and $B$ configurations, and thus $A$ and $B$ are the only possible
such configurations.
We call these \emph{nondegenerate} reversible hinged dissections.
(For example, Figure~\ref{dudeney} is nondegenerate, while
Figure~\ref{nonmon} is degenerate and not monotone, and
Figure~\ref{degen} is degenerate and monotone.)

\begin{figure}
  $$
  \vcenter{\hbox{\includegraphics[width=0.33\linewidth]{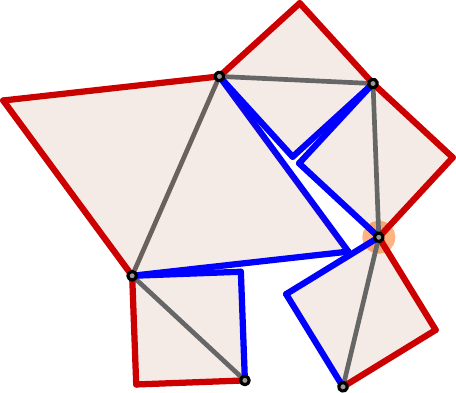}}}
  \quad\leftrightarrow\quad
  \vcenter{\hbox{\includegraphics[width=0.33\linewidth]{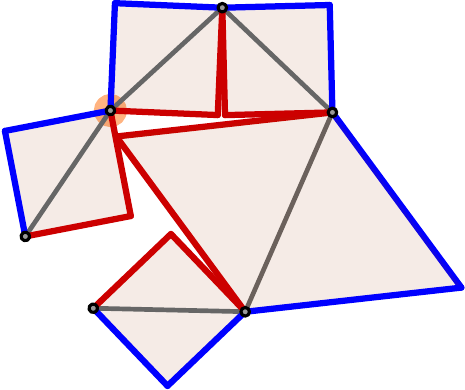}}}
  $$
  \caption{Reversible hinged dissection that is monotone but degenerate,
           because of the highlighted vertex.}
  \label{degen}
\end{figure}

\begin{lemma} \label{lem:nondegenerate implies monotone}
Every nondegenerate reversible hinged dissection is strictly monotone,
i.e., monotone and all turn angles change between $A$ and~$B$.
\end{lemma}
\iffull
\begin{proof}
Pick the reference points $v^-$ and $v^+$  and define angle
$\alpha_v$ and $\alpha'_v$ as in Theorem~\ref{thm:monotone}.
Because the dissection is
nondegenerate, the two pieces attached to hinge $v$ touch on the inside of
$A$. Therefore, for any hinge angle $\angle v^-,v,v^+$ less than $\alpha_v$,
those two pieces would intersect. Because no two piece intersect when
the dissection forms polygon $B$, $\alpha'_v \geq \alpha_v$ for all
$v$ and the dissection is monotone.
Furthermore, we cannot have $\alpha'_v = \alpha_v$, because then the two
pieces would touch on both sides of the hinge, meaning that it is not a
hinge at all.
\end{proof}
\fi

The distinguishing feature of nondegenerate reversible hinged dissections is that
their corresponding unfoldings are \emph{nondegenerate} in the sense that every
vertex of the polyhedron has just a single incident cut (degree~$1$)
in the cut tree.

\begin{theorem} \label{thm:nondegenerate}
Two polygons $A$ and $B$ have a nondegenerate reversible hinged dissection
if and only if $A$ and $B$ are two nondegenerate noncrossing nets of a common convex 
polyhedron.
\end{theorem}

\begin{proof}
By Lemma~\ref{lem:nondegenerate implies monotone}, a nondegenerate reversible hinged
dissection between $A$ and $B$ is strictly monotone, so by
Theorem~\ref{thm:monotone}, $A$ and $B$ are two noncrossing nets of a
common convex polyhedron.
When gluing the pieces along both dissection boundaries to form the
convex polyhedral metric, we form vertices exactly at the hinges
(by strict monotonicity).
By nondegeneracy, no other piece touches a hinge, so the resulting polyhedron
vertex has only one cut in each of the two unfoldings (corresponding to the
opened edge in each hinged dissection).

In the other direction, if $A$ and $B$ are two nondegenerate noncrossing nets of a
common convex polyhedron, then by Theorem~\ref{thm:monotone},
they have a monotone reversible hinged dissection between $A$ and~$B$.
Each of the two states of the hinged dissection is formed by cutting
the cut tree of one unfolding (e.g., $A$) while leaving attached the cut tree
of the other unfolding (e.g., $B$).  Because both unfoldings are nondegenerate,
at each hinge, the pieces share one side (e.g., $B$) while leaving an empty
sector angle on the other side (e.g., $A$), so no other piece can be incident
to the hinge.  Therefore the unfolding is nondegenerate.
\end{proof}


\iffull
\fi

Figure~\ref{examples} shows two examples of hinged dissections resulting
from these techniques.  Historically, many hinged dissections
(e.g., in \cite{Frederickson-1997,Frederickson-2002}) have been designed
by overlaying tessellations of the plane by shapes $A$ and~$B$.
This connection to tiling is formalized by the results of this paper,
combined with the characterization of shapes that tile the plane isohedrally
as unfoldings of certain convex polyhedra \cite{tilemakers}.

\begin{figure*}
  \centering
  \subcaptionbox{Two noncrossing nets of a doubly covered triangle.}
    {\includegraphics[scale=0.9]{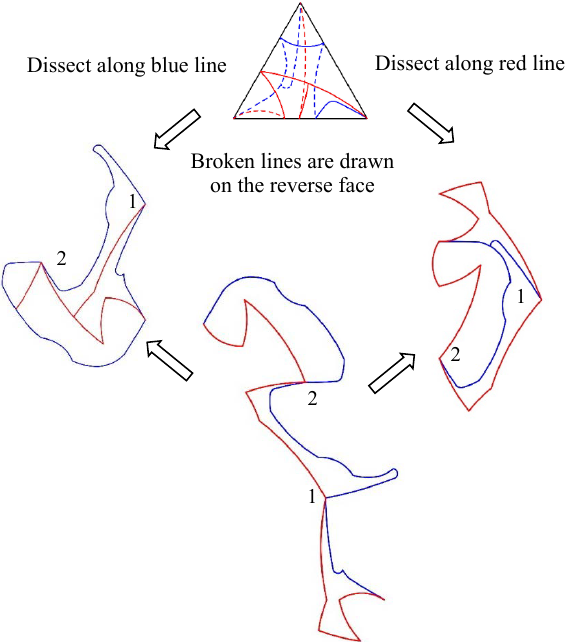}}\hfill
  \subcaptionbox{Lobster to fish: two noncrossing nets of a doubly covered rectangle.}{\includegraphics[scale=0.98]{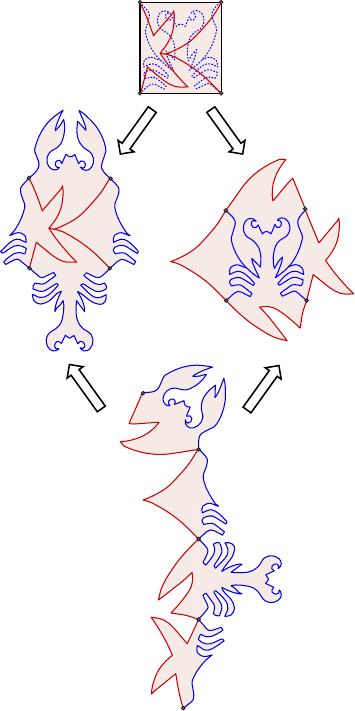}}
  \caption{Two nondegenerate reversible hinged dissections found by parcel
    and envelope magic.}
  \label{examples}
\end{figure*}

\let\realbibitem=\bibitem
\def\bibitem{\par \vspace{-1ex}\realbibitem}

\small
\bibliographystyle{plain}
\bibliography{paper}

\begin{thebibliography}{10}

\bibitem{Abbott-Abel-Charlton-Demaine-Demaine-Kominers-2012}
Timothy~G. Abbott, Zachary Abel, David Charlton, Erik~D. Demaine, Martin~L.
  Demaine, and Scott~Duke Kominers.
\newblock Hinged dissections exist.
\newblock {\em Discrete \& Computational Geometry}, 47(1):150--186, 2012.

\bibitem{Akiyama-Langerman-Matsunaga-2015}
Jin Akiyama, Stefan Langerman, and Kiyoko Matsunaga.
\newblock Reversible nets of polyhedra.
\newblock In {\em Revised Papers from the 18th Japan Conference on Discrete and
  Computational Geometry and Graphs}, pages 13--23, Kyoto, Japan, September
  2015.

\bibitem{Akiyama-Nakamura-1998}
Jin Akiyama and Gisaku Nakamura.
\newblock Dudeney dissection of polygons.
\newblock In {\em Revised Papers from the Japan Conference on Discrete and
  Computational Geometry}, volume 1763 of {\em Lecture Notes in Computer
  Science}, pages 14--29, Tokyo, Japan, December 1998.

\bibitem{Akiyama-Seong-2013}
Jin Akiyama and Hyunwoo Seong.
\newblock An algorithm for determining whether a pair of polygons is
  reversible.
\newblock In {\em Proceedings of the 3rd Joint International Conference on
  Frontiers in Algorithmics and Algorithmic Aspects in Information and
  Management}, pages 2--3, Dalian, China, June 2013.

\bibitem{Alexandrov-2005}
A.~D. Alexandrov.
\newblock {\em Convex Polyhedra}.
\newblock Springer, 2005.
\newblock Translation of Russian edition, 1950.

\bibitem{BZ96}
Yu.~D. Burago and V.~A. Zalgaller.
\newblock Isometric piecewise linear immersions of two-dimensional manifolds
  with polyhedral metrics in {$\mathbb{R}^3$}.
\newblock {\em St. Petersburg Math. J.}, 7:369--385, 1996.
\newblock Translation by S.V. Ivanov.

\bibitem{Demaine-O'Rourke-2007}
Erik~D. Demaine and Joseph O'Rourke.
\newblock {\em Geometric Folding Algorithms: Linkages, Origami, Polyhedra}.
\newblock Cambridge University Press, July 2007.

\bibitem{Dudeney-1902-hinged}
Henry~E. Dudeney.
\newblock Puzzles and prizes.
\newblock {\em Weekly Dispatch}, 1902.
\newblock The puzzle appeared in the April 6 issue of this column. A discussion
  followed on April 20, and the solution appeared on May 4.

\bibitem{Frederickson-1997}
Greg~N. Frederickson.
\newblock {\em Dissections: Plane and Fancy}.
\newblock Cambridge University Press, November 1997.

\bibitem{Frederickson-2002}
Greg~N. Frederickson.
\newblock {\em Hinged Dissections: Swinging \& Twisting}.
\newblock Cambridge University Press, August 2002.

\bibitem{tilemakers}
Stefan Langerman and Andrew Winslow.
\newblock A complete classification of tile-makers.
\newblock In {\em Abstracts from the 18th Japan Conference on Discrete and
  Computational Geometry and Graphs}, Kyoto, Japan, September 2015.

\bibitem{orourke2010folding}
Joseph O'Rourke.
\newblock On folding a polygon to a polyhedron.
\newblock arXiv:1007.3181, July 2010.
\newblock \url{http://arXiv.org/abs/1007.3181}.

\bibitem{Uehara-2016-boxes}
Ryuhei Uehara.
\newblock A survey and recent results about common developments of two or more
  boxes.
\newblock In {\em Origami$^6$: Proceedings of the 6th International Meeting on
  Origami in Science, Mathematics and Education}, volume 1 (Mathematics), pages
  77--84. American Mathematical Society, 2014.

\bibitem{BZ60}
Ю.~Д. Бураго and В.~А. Залгаллер.
\newblock Реализация разверток в виде
  многогранников (polyhedral embedding of a net).
\newblock {\em Вестник ЛГУ}, 15:66--80, 1960.
\newblock In Russian with English summary.

\bibitem{BZ95}
Ю.~Д. Бураго and В.~А. Залгаллер.
\newblock Изометрические кусочно-линейные
  погружения двумерных многообразий
  с~полиэдральной метрикой в~{$\mathbb{R}^3$}.
\newblock {\em Алгебра и анализ}, 7:76--95, 1995.
\newblock In Russian, translation in \cite{BZ96}.

\end{thebibliography}

\end{document}